%% file: paper.tex
\documentclass[a4paper]{article}
\usepackage{amsmath}
\usepackage{amsthm}
\usepackage{url}
\usepackage{adjustbox}
\usepackage{algpseudocode}
\usepackage{algorithm}
\usepackage{pstricks,pst-node}
\usepackage[square,comma,sort]{natbib}

\def\bhBox#1#2#3{\parbox[c][#1][c]{#2}{
    \makebox[#2]{#3}}}
\def\sItem#1{
  \psframebox[framearc=0.7,fillstyle=solid]{\bhBox{10pt}{35pt}{$#1$}}
}

\newtheorem{theorem}{Theorem}
\newtheorem{lemma}{Lemma}
\newtheorem{corolary}{Corolary}

\author{Lu\'is M. S. Russo\\
  \texttt{luis.russo@tecnico.ulisboa.pt}\\
  INESC-ID and Department of Computer Science and Engineering,\\
  Instituto Superior T\'{e}cnico, Universidade de Lisboa.}

\title{Range Minimum Queries in Minimal Space}

\date{}

\begin{document}

\maketitle

\begin{abstract}
  We consider the problem of computing a sequence of range minimum
  queries. We assume a sequence of commands that contains values and
  queries. Our goal is to quickly determine the minimum value that exists
  between the current position and a previous position $i$. Range minimum
  queries are used as a sub-routine of several algorithms, namely related
  to string processing. We propose a data structure that can process these
  commands sequences. We obtain efficient results for several variations of
  the problem, in particular we obtain $O(1)$ time per command for the
  offline version and $O(\alpha(n))$ amortized time for the online version,
  where $\alpha(n)$ is the inverse Ackermann function and $n$ the number of
  values in the sequence. This data structure also has very small space
  requirements, namely $O(\ell)$ where $\ell$ is the maximum number active
  $i$ positions. We implemented our data structure and show that it is
  competitive against existing alternatives. We obtain comparable command
  processing time, in the nano second range, and much smaller space
  requirements.
\end{abstract}
{\bf Keywords:} Range Minimum Queries, Union Find, Disjoint Sets,  Bulk
Queries, String Processing, Longest Common Extension.

\section{The Problem}
\label{sec:problem}

Given a sequence of integers, usually stored in an array $A$, a range
minimum query (RMQ) is a pair of indexes $(i,j)$. We assume that
$i \leq j$. The solution to the query consists finding in the minimum value
that occurs in $A$ between the indexes $i$ and $j$. Formaly the solution is
$\min \{A[k] | i \leq k \leq j\}$. There exist several efficient solutions
for this problem, in this static offline context, see
Section~\ref{sec:related-work}. In this paper we consider the case where
$A$ is not necessarially stored. Instead we assume that the elements of $A$
are streamed in a sequential fashion. Likewise we assume that the
corresponding queries and are intermixed with the values of $A$ and the
answers to the operations are computed online. Hence we assume that the
input to our algorithm consists in a sequence of the following commands:
\begin{description}
\item[Value] - represented by \texttt{V}, is followed by an integer, or
  float, value $v$ and it indicates that $v$ is the next entry of $A$,
  i.e., $A[j] = v$.
\item[Query] - represented by \texttt{Q}, is followed by an integer that
  indicates a previous index of the sequence. The given integer corresponds
  to the element $i$ in the query. The element $j$ is the position of the
  last given value of $A$. Hence it is only necessary to specify $i$. This
  command can only be issued if an \texttt{M} command was given at position
  $i$ and no close command was given with argument $i$.
\item[Mark] - represented by \texttt{M}, indicates that future queries may
  use the current position $j$ as element $i$, i.e., as the beginning of
  the query.
\item[Close] - represented by \texttt{C}, is also followed by an integer
  $i$ that represents an index of the sequence. This command essentially
  nullifies the effect of an \texttt{M} command issued at position $i$.
  Hence the command indicates that the input contains no more queries that
  use $i$. Any information that is being kept about position $i$ can be
  purged.
\end{description}
For simplicity we assume that the sequence of commands is not designed to
hack our data structure. Hence we assume that no patological sequences are
given as input. Examples of patological sequences would be: issuing the
\texttt{Mark} command twice or more or mixed with \texttt{Query}; issuing a
\texttt{Close} command for an index that was not marked; issuing
\texttt{Mark} commands for positions that have been closed; etc.

Consider the following example sequence. We will use this sequence
throughout the paper.

\begin{verbatim}
V 22 M V 23 M V 26 M V 28 M V 32 M V 27 M V 35 M Q 4 C 3
\end{verbatim}

In this paper we study this type of sequences. Our contributions are the
following:
\begin{itemize}
\item We propose a new algorithm that can efficiently process this type of
  input sequences. We show that our algorithm produces the correct
  solution.
\item We analyze the algorithm and show that it obtains a fast running time
  and requires only a very small amount of space. Specifically the space
  requirements are shown to be at most $O(q)$, where $q$ is the number of
  queries. Recall that we do not store the array $A$. We further reduce
  this bound to $O(\ell)$. Consider at some instant the number of marked
  positions that have not yet been closed. We refer to these positions as
  active. The maximum number of active positions over all instants is
  $\ell$. The query time is shown to be $O(1)$ in the offline version of
  the problem and $O(\alpha(\ell))$ on the online version, where $\alpha$
  is the inverse Ackermann function, see Theorem~\ref{teo:basic} and
  Corolary~\ref{cor:const} in Section~\ref{sec:analysis}. We also discuss
  the use of this data structure for real-time applications. We obtain a
  high probability $O(\log n)$ time for all operations,
  Theorem~\ref{teo:high}. We also discuss trade-off that can reduce this
  bound to $O(\log \log n)$ for some operations, Theorem~\ref{teo:real}.
\item We implemented the online version of our algorithm and show
  experimentally that it is very efficient both in time and space.
\end{itemize}
\section{Data Structure Outline}
\label{sec:solution}
Let us now dicuss how to solve this problem, by gradually considering the
challenge at hand. We start by describing a simple structure. We then
proceed to improve its performance, first by selecting fast data structures
which provide good time bounds and second by reducing the space
requirements from $O(q)$ to $O(\ell)$.

Consider again the sequence in Section~\ref{sec:problem}. Our first data
structure is a stack, which we use in the same way as for building a
Cartesian tree, see~\citet*{Crochemore_2020}. The process is simple. We
start by pushing a $-\infty$ value into the stack, this value will be used
as a sentinel. To start the discussion we will assume, for now, that every
\texttt{Value} command is followed by a \texttt{Mark} command, meaning that
every position is relevant for future queries.

An important invariant of this stack is that the values form an increasing
sequence. Whenever a value is received it is compared with the top of the
stack. While the value at hand is smaller the stack gets poped. At some
point the input value will be larger than the top of the stack, even if it
is necessary for the sentinel to reach the top. When the input value is
larger than the top value it gets pushed into the stack. Another important
property of this data structure is that the values in the stack are the
only possible solutions for range minimum queries $(i,j)$, where $j$ is the
current position of the sequence being processed and $i$ is some previous
position.

To identify the corresponding $i$ it is usefull to keep, associated to each
stack item, the set of positions that yield the corresponding item as the
RMQ solution. Maintaining this set of positions is fairly simple. Whenever
an item is inserted into the stack it is inserted with the current
position. We number positions by starting at $1$. When an item is poped
from the stack the set of positions associated to that item is transferred
into the set of positions of the item below it. In our example the
\texttt{Value} $27$ command puts the positions $4$ and $5$ into the same
set. The rightmost gray rectangle in Figure~\ref{fig:stack1} illustrates
the state of this data structure after processing the commands \texttt{V 35
  M} of our sample sequence. To process a \texttt{Close} command we remove
the corresponding position from whatever set it belongs to, i.e., command
\texttt{C} followed by $i$ removes $i$ from a position set.

Figure~\ref{fig:stack1} illustrates the configuration of this data
structure as it processes the following sequence of commands:
\begin{verbatim}
V 22 M V 23 M V 26 M V 28 M V 32 M V 27 M V 35 M Q 4 C 3
\end{verbatim}
Each gray rectangle shows a different configuration. The leftmost
configuration is obtained after the \texttt{V 32 M} commands. The second
configuration after the \texttt{V 35 M} commands. The rightmost
configuration is the final one after the \texttt{C 3}. The solution to the
\texttt{Q 4} command is $27$, because it is the stack item associated with
the position $4$ in the rightmost configuration, these values are
highlighted in bold.

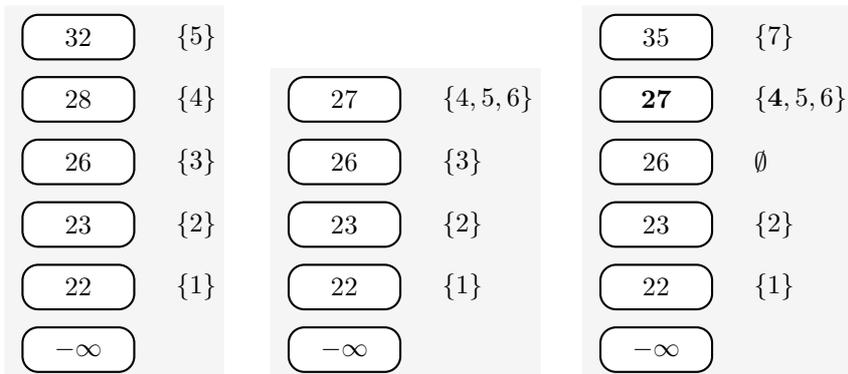
\begin{figure}[tb]
  \centering
  \begin{pspicture}(11.5,5)
    \rput[bl](0,0){
      \psframebox*[fillstyle=solid,fillcolor=black!4]{
        \begin{psmatrix}[mnode=r,rowsep=5pt,colsep=12pt,mcol=l]
          \sItem{32} & $\{5\}$ \\
          \sItem{28} & $\{4\}$ \\
          \sItem{26} & $\{3\}$ \\
          \sItem{23} & $\{2\}$ \\
          \sItem{22} & $\{1\}$ \\
          \sItem{-\infty}
        \end{psmatrix}}}
    \rput[bl](3.5,0){
      \psframebox*[fillstyle=solid,fillcolor=black!4]{
        \begin{psmatrix}[mnode=r,rowsep=5pt,colsep=12pt,mcol=l]
          \sItem{27} & $\{4,5,6\}$ \\
          \sItem{26} & $\{3\}$ \\
          \sItem{23} & $\{2\}$ \\
          \sItem{22} & $\{1\}$ \\
          \sItem{-\infty}
        \end{psmatrix}}}
    \rput[bl](7.6,0){
      \psframebox*[fillstyle=solid,fillcolor=black!4]{
        \begin{psmatrix}[mnode=r,rowsep=5pt,colsep=12pt,mcol=l]
          \sItem{35} & $\{7\}$ \\
          \sItem{\mathbf{27}} & $\{\mathbf{4},5,6\}$ \\
          \sItem{26} & $\emptyset$ \\
          \sItem{23} & $\{2\}$ \\
          \sItem{22} & $\{1\}$ \\
          \sItem{-\infty}
        \end{psmatrix}}}
  \end{pspicture}
  \caption[Structure illustration]{Illustration of structure configuration
    at different instances.  Each gray rectangle shows the stack on the
    left and the corresponding sets of positions on the right. }
  \label{fig:stack1}
\end{figure}

Using a standard stack implementation it is possible to guarantee $O(1)$
time for the push and pop operations. Hence, ignoring the time required to
process the sets of positions, the pairs of \texttt{Value} and
\texttt{Mark} operations require only constant amortized time to
compute. In the worst case a \texttt{Value} operation may need to discard a
big stack, i.e., it may require poping $O(n)$ items, where $n$ is the total
amount of positions in $A$. However since each operation executes at most
one push operation the amortized time becomes $O(1)$. Hence the main
challenge for this data structure is how to represent the sets of
positions. To answer this question we must first consider how to compute
the \texttt{Query} operation. Given this command, followed by a value $i$,
we proceed to find the set that contains $i$ and report the corresponding
stack element. For example to process the \texttt{Q} $4$ command in the
input sequence we most locate the set that contains position $4$. In this
case the set is $\{4, 5, 6\}$ and the corresponding element is $27$. Hence
the essential operations that are required for the sets of positions are
the union and the find operations. Union is used when merging sets in the
\texttt{Mark} operation and find is used to identify sets in the
\texttt{Query} operation.

A naive implementation requires $O(n)$ time for each operation. Instead we
use a dedicated data structure that supports both operations in
$O(\alpha(n))$ amortized time, where $\alpha(n)$ is the inverse Ackermann
function. Note that although conceptually the \texttt{Close} command
removes elements from the position sets this data structure is essentially
ignoring these operations. They do not alter the Union-Find (UF) data
structure. Hence, once an element is assigned to a set, it can no longer be
removed. Fortunately the resulting procedure is still sound, albeit it
requires more space. This version does require a large amount of space,
specifically $O(n)$ space.

Let us now focus on reducing the space to $O(m)$, where $m$ is the total
number of \texttt{Mark} commands, which should be equal to the total number
of \texttt{Close} commands. We must also have that $m \leq q$, where $q$ is
the number of \texttt{Query} commands, as there is no point in issuing
redundant \texttt{Mark} commands. Note that $m$ may be much smaller than
$n$ as there might be many more \texttt{Value} commands than \texttt{Mark}
commands.

To guarantee that the size of the stack is at most $O(m)$ we now consider
the situation where not all the \texttt{Value} commands are followed by
\texttt{Mark} commands, otherwise $n$ and $m$ would be similar. In this
case only the marked positions need to be stored in the stack, thus
reducing its size. This separation of commands means that our operating
procedure also gets divided. The \texttt{Mark} command only pushes elements
into the stack. The \texttt{Value} commands only performs the poping
commands. Hence in this scenario both the \texttt{Mark} and \texttt{Value}
commands require $O(\alpha(n))$ amortized time.

To illustrate the division we have just described consider the following
sequence of commands:
\begin{verbatim}
V 22 M V 23 V 26 M V 28 M V 32 M V 27 M V 35 M Q 4 C 3
\end{verbatim}
We illustrate the state of the resulting data structure in
Figure~\ref{fig:stack2}. Notice that in this sequence there is no
\texttt{M} command after \texttt{V 23}. Therefore this value never gets
inserted into the stack.
\begin{figure}[tb]
  \centering
  \begin{pspicture}(11.5,4.3)
    \rput[bl](0,0){
      \psframebox*[fillstyle=solid,fillcolor=black!4]{
        \begin{psmatrix}[mnode=r,rowsep=5pt,colsep=12pt,mcol=l]
          \sItem{32} & $\{5\}$ \\
          \sItem{28} & $\{4\}$ \\
          \sItem{26} & $\{3\}$ \\
          \sItem{22} & $\{1\}$ \\
          \sItem{-\infty}
        \end{psmatrix}}}
    \rput[bl](3.5,0){
      \psframebox*[fillstyle=solid,fillcolor=black!4]{
        \begin{psmatrix}[mnode=r,rowsep=5pt,colsep=12pt,mcol=l]
          \sItem{27} & $\{4,5,6\}$ \\
          \sItem{26} & $\{3\}$ \\
          \sItem{22} & $\{1\}$ \\
          \sItem{-\infty}
        \end{psmatrix}}}
    \rput[bl](7.6,0){
      \psframebox*[fillstyle=solid,fillcolor=black!4]{
        \begin{psmatrix}[mnode=r,rowsep=5pt,colsep=12pt,mcol=l]
          \sItem{35} & $\{7\}$ \\
          \sItem{27} & $\{4,5,6\}$ \\
          \sItem{26} & $\emptyset$ \\
          \sItem{22} & $\{1\}$ \\
          \sItem{-\infty}
        \end{psmatrix}}}
  \end{pspicture}
  \caption[Omitted \texttt{Mark} illustration]{Illustration of structure
    configuration at different instances. In this sequence of commands
    there is no \texttt{M} command after \texttt{V 23}. Each gray rectangle
    shows the stack on the left and the corresponding sets of positions on
    the right.}
  \label{fig:stack2}
\end{figure}
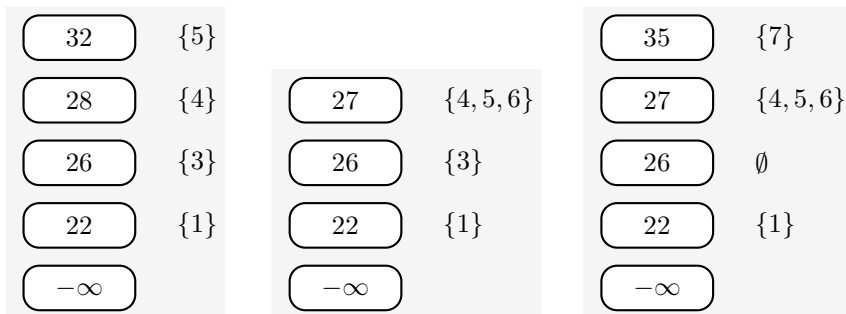

To reduce the size of the UF data structure we add a hash table to
it. Without this table every one of the $n$ position values are elements
for the UF data structure. Using a hash we can filter out only the marked
positions. When a \texttt{Mark} command is issued we insert the current $j$
position as the hash key and the value is the current number of UF
elements. This reduces the size of the UF data structure to $O(m)$.
Moreover the hash table also requires only $O(m)$ space. Hence this data
structure requires only $O(m)$ space and can process any sequence of
commands in at most $O(\alpha(n))$ amortized time per command.  When a
\texttt{Close} $i$ command is issued we mark the position $i$ as deleted in
the hash table, but we do not actually remove it from memory. The reason
for this process is that a stack item might actually point to position $i$
and removing it would break the data structure. For the $O(m)$ space bound
this is not an issue as inactive markings count for the overall total.

In the next section we discuss several nuances of this data structure,
including how to further reduce the space requirements to $O(\ell)$ space
and alternative implementations.

\section{The Details}
\label{sec:details}
In this Section we will prove that the algorithm is correct and analyze its
performance. We start of by giving a pseudo code description of the
algorithms used for each command,
Algorithms~\ref{alg:query},~\ref{alg:close},~\ref{alg:value}
and~\ref{alg:Mark}. In these algorithms we make some simplifying
assumptions and use some extra commands that we will now define.

For simplicity we describe the data structure that does not use a
hash-table. We use $S$ to represent the stack data structure, but we also
use $S[k']$ to reference the element at position $k'$. In general the top
of the stack is at position $k$, which also corresponds to the number of
elements in the stack. We use $k$ as a global variable. We also use $k$ as
a bounded variable in the Lemma statements. Hence the value of $k$ must be
derived from context. This is usually not a problem and in fact it is handy
for the proofs, which most of the time only need to consider when $k$ is
the top of the stack. We also use the notation \texttt{Top}$(S)$ to refer
to the top of the stack, this value is equal to $S[k]$. Note that this
means that the element $S[k-1]$ is the one just below the \texttt{Top}
element. Algorithms~\ref{alg:push} and~\ref{alg:pop} used to manipulate the
stack status and are given for completion. The set of positions associated
with each stack item are denoted with the letter $P$.  In our example we
have that $P[4] = \{4, 5, 6\}$, see Figure~\ref{fig:stack1}.

In algorithm~\ref{alg:query} we assume that the result of the \texttt{Find}
command is directly a position index of $S$, therefore the expression
$S[\texttt{Find}(i)]$ for Algorithm~\ref{alg:query}. The \texttt{NOP}
command does nothing, it is used to highlight that without a hash table
there is nothing for the \texttt{Close} command to execute.

The \texttt{Make-Set} function is used to create a set in the UF data
structure, the first argument indicates the element that is stored in the
set (position $j$) and the second argument the level of the last element on
the stack $S$, i.e., $k$. It is the values given in this second argument
that we expect \texttt{Find} to return. Likewise the \texttt{Union}
function receives three arguments. The sets that we want to unite and again
the top of the stack $k$. Note that in Algorithm~\ref{alg:Mark} we use
$\{j\}$ as one of the arguments to \texttt{Union} operation. In this case
we are assuming that this operation makes the corresponding
\texttt{Make-Set} operation.

Besides $k$ we have a few global variables, $j$ which indicates the current
position in $A$ and $v$, which is not an argument of the \texttt{Mark}
command but is used in that command. At that point it is assumed that $v$
is the last value given in the \texttt{Value} command.

\begin{algorithm}
  \caption{} \label{alg:push}
  \begin{algorithmic}[1]
    \Procedure{Push}{$v$}\Comment{Insert element in stack}
    \State $k\gets k+1$
    \State $S[k] = v$
    \EndProcedure
  \end{algorithmic}
\end{algorithm}
\begin{algorithm}
  \caption{} \label{alg:pop}
  \begin{algorithmic}[1]
    \Procedure{Pop}{}\Comment{Remove element from stack}
    \State $k\gets k-1$
    \EndProcedure
  \end{algorithmic}
\end{algorithm}
\begin{algorithm}
  \caption{} \label{alg:query}
  \begin{algorithmic}[1]
    \Procedure{Query}{$i$}\Comment{Return RMQ $(i,j)$}
    \State \Return $S[\texttt{Find}(i)]$
    \EndProcedure
  \end{algorithmic}
\end{algorithm}
\begin{algorithm}
  \caption{} \label{alg:close}
  \begin{algorithmic}[1]
    \Procedure{Close}{$i$}\Comment{Ignore command}
    \State \texttt{NOP}
    \EndProcedure
  \end{algorithmic}
\end{algorithm}
\begin{algorithm}
  \caption{} \label{alg:value}
  \begin{algorithmic}[1]
    \Procedure{Value}{$v$}\Comment{Put into the stack}
    \If {$S[k]> v$} \Comment{Test element at the
      \texttt{Top}.} \label{line:wif}
    \While {$S[k-1] \geq v$} \Comment{Test element below the
      \texttt{Top}.} \label{line:wguard}
    \State \texttt{Union}(\texttt{P}$[k]$, \texttt{P}$[k-1]$, $k$) \label{line:union}
    \Comment{Unite top position sets.}
    \State \texttt{Pop}() \label{line:pop}
    \EndWhile
    \State $S[k]= v$ \label{line:setTop}
    \EndIf
    \State $j\gets j+1$
    \EndProcedure
  \end{algorithmic}
\end{algorithm}
\begin{algorithm}
  \caption{} \label{alg:Mark}
  \begin{algorithmic}[1]
    \Procedure{Mark}{}\Comment{Put into the stack}
    \If {$S[k]< v$} \label{line:mif}
    \State \texttt{Push}($v$) \Comment{Insert $v$ into $S$.} \label{line:push}
    \State \texttt{Make-Set}$(j, k)$ \Comment{Associate with $k$.}
    \Else
    \State \texttt{Union}(\texttt{P}$[k]$, $\{j\}$, $k$) \Comment{Assume it
      calls \texttt{Make-Set}.}
    \EndIf
    \EndProcedure
  \end{algorithmic}
\end{algorithm}
\subsection{Correctness}
\label{sec:correctness}
In this Section we establish that our algorithm is correct, meaning the
values obtained from our data structure actually correspond to the
solutions of the given range minimum queries. We state several invariant
properties that the structure always maintains.

We consider the version of the data structure that consists of a stack and
a UF structure. The version containing a hash is relevant for obtaining an
efficient structure but does not alter the underlying operation
logic. Hence the correctness of the algorithm is preserved, only its
description is more elaborate.

We prove the invariant properties by structural induction, meaning that we
assume that they are true before a command is processed and only need to
prove that the property is maintained by the corresponding processing. For
this kind of argument to hold it is necessary to verify that the given
properties are also true when the structure is initialized, this is in
general trivially true so we omit this verification from the following
proofs. Another declutering observation is that the \texttt{Query} and
\texttt{Close} commands do not alter our data structure and therefore are
also omitted from the following proofs.

Let us start by establishing some simple properties.
\begin{lemma}
  The stack $S$ always contains at least two elements.
  \label{lemma:stwo}
\end{lemma}
\begin{proof}
  In this particular proof it is relevant to mention the initial state of
  the stack $S$. The stack is initialized with two sentinel values,
  $-\infty$ followed by $+\infty$. Hence it initially contains at least two
  elements.
  \begin{itemize}
  \item The \texttt{Mark} command. This command does not uses the
    \texttt{Pop} operation and therefore never reduces the number of
    elements. The result follows by induction hypothesis.

  \item The \texttt{Value} command. For the \texttt{Pop} operation in
    line~\ref{line:pop} of Algorithm~\ref{alg:value} to execute the
    \textbf{while} guard in line~\ref{line:wguard} must be true. Note that
    when $k = 2$ this guard consists in testing whether $-\infty = S[1] >
    v$, which is never the case and therefore a \texttt{Pop} operation is
    never executed in a stack that contains $2$ elements.
  \end{itemize}
\end{proof}

\begin{lemma}
  If $v$ was the argument of the last \texttt{Value} command and $k$ is the
  top level of that stack $S$ then $S[k] \leq v$.
  \label{lemma:top}
\end{lemma}
\begin{proof}$ $\newline
  \begin{itemize}
  \item The \texttt{Mark} command. When the \textbf{if} condition of
    Algorithm~\ref{alg:Mark} is true we have that line~\ref{line:push}
    executes. After which $S[k] = v$ and the Lemma condition is
    verified. Otherwise the \textbf{if} condition is false and the stack is
    kept unaltered, in which case the result follows by induction
    hypothesis.

  \item The \texttt{Value} command. When the \textbf{if} condition of
    Algorithm~\ref{alg:value} fails the Lemma property is immediate. Hence
    we only need to check the case when the \textbf{if} condition holds. In
    this case line~\ref{line:setTop} must eventually execute at which point
    we have that $S[k] = v$ and the Lemma condition is verified.
  \end{itemize}
\end{proof}

Let us now focus on more global properties. Next we show that the values
stored in $S$ are in increasing order.

\begin{lemma}
  For any indexes $k$ and $k'$ of the stack $S$ we have that if $k' < k$
  then $S[k'] < S[k]$.
  \label{lemma:sortS1}
\end{lemma}
\begin{proof}$ $\newline
  \begin{itemize}
  \item The \texttt{Value} command. This command does not push elements
    into the stack, instead it pops elements. This means that, in general,
    a few relations are discarded. The remaining relations are preserved by
    the induction hypothesis. The only change that we need to verify is if
    the \texttt{Top} of the stack $S$ changes, line~\ref{line:setTop} of
    Algorithm~\ref{alg:value}. Hence we need to check the case when $k$ is
    the top level of the stack. Note that line~\ref{line:setTop} occurs
    immediately after the \textbf{while} cycle.  Which means that its guard
    is false, i.e., we have that $S[k-1] < v = S[k]$. Hence the desired
    property was established for $k' = k-1$. For any other $k' < k-1$ we
    can use the induction hypothesis to conclude that $S[k'] < S[k-1]$,
    which combined with the previous inequality and transitivity yields the
    desired property that $S[k'] < S[k]$.
  \item The \texttt{Mark} command. The only operation performed by this
    command is to push the last element into the stack. Hence when $k$ is
    below the top of the stack the property holds by induction. Let us
    analyze the case when the top of the stack changes, i.e., when $k$ is
    the top level of the stack. The change occurs in line~\ref{line:push}
    of Algorithm~\ref{alg:Mark} in which case we have that $S[k-1] < v =
    S[k]$. Hence we extend the argument for $k' < k-1$ as in the
    \texttt{Value} command by induction hypothesis and transitivity.
  \end{itemize}
\end{proof}
Likewise the converse of this Lemma can now be established.
\begin{lemma}
  For any indexes $k$ and $k'$ of the stack $S$ we have that if $S[k] <
  S[k']$ then $k < k'$.
  \label{lemma:sortS2}
\end{lemma}
\begin{proof}
  Assume by contradiction that there are $k$ and $k'$ such that
  $S[k] < S[k']$ and $k' \leq k$. Because $S[k] \neq S[k']$ we have that
  $k \neq k'$, since we are using $S$ as an array. Hence we must have that
  $k' < k$ and can now apply Lemma~\ref{lemma:sortS1} to conclude that
  $S[k'] < S[k]$, which contradicts the order relation in our hypothesis.
\end{proof}
This sorted property also gives structure to the sets of positions.
\begin{lemma}
  For any indexes $k' < k$ and positions $p' \in P[k']$ and $p \in P[k]$
  we have that $p' < p$.
  \label{lemma:sortP}
\end{lemma}
\begin{proof}$ $\newline
  \begin{itemize}
  \item The \texttt{Mark} command. This operation inserts the current
    position $j$ into the set that corresponds to the top of the stack.
    The top might have been preserved or created by the operation, both
    cases can be justified in the same way. We only need to consider the
    case when \texttt{Top}$(S) = S[k]$ and $p = j$, any other
    instanciation of the variables in the Lemma will correspond to
    relations that were established before the structure was
    modified. Hence we only need to show that $p' < j$ for any $p'$ in any
    $P[k']$. This is trivial because $j$ represents the current position in
    $A$, which is therefore larger than any previous position of $A$ that
    may be represented by $p'$.
  \item The \texttt{Value} command. As this command pops elements from the
    stack, it has the side effect of merging the position sets. Hence the
    only new relation is for positions at the top of the stack, i.e., when
    $p \in P[k]$ and \texttt{Top}$(S) = S[k]$. We only need to consider
    where position $p$ was before the operation, i.e., $p \in P_b[k_b]$,
    were $P_b[k_b]$ represents a set of positions before the operation is
    executed. Because the \texttt{Value} command merges the position sets
    which are highest on the stack we have that $k \leq k_b$. Now, for any
    $k' < k$ and $p' \in P[k']$, we have that $P[k'] = P_b[k']$ because the
    sets of positions below the top of the stack are not altered by the
    operation. In essence we have that $k' < k_b$ and $p' \in P_b[k']$ and $p
    \in P_b[k_b]$, therefore by induction hypothesis we obtain $p' < p$, as
    desired.
  \end{itemize}
\end{proof}
We can now state our final invariant, which establishes that our algorithm
is correct.
\begin{theorem}
  At any given instant when $j$ is the current position over $A$ we have
  that if $i \in P[k']$ then RMQ$(i,j) = S[k']$.
  \label{teo:correct}
\end{theorem}
\begin{proof}$ $\newline
  \begin{itemize}
  \item The \texttt{Mark} command. This command does not alter the sequence
    $A$. Therefore none of the RMQ$(i,j)$ values change. Since almost all
    positions and position sets $P[k']$ are preserved the implication is
    also preserved. The only new position is $j \in P[k]$, therefore the
    only case we need to consider is when $i = j$ and $k'$ is the top level
    of the stack $S$, i.e., $k' = k$. In this case we have that
    RMQ$(j,j) = A[j] = v$, where $v$ is the argument given in the last
    \texttt{Value} command. Now let us consider the \textbf{if} condition
    in line~\ref{line:mif} of Algorithm~\ref{alg:Mark}. This further
    divides the argument into two cases:
    \begin{itemize}
    \item When this condition holds then line~\ref{line:push} of
      Algorithm~\ref{alg:Mark} executes and makes $S[k] = v$. Hence
      RMQ$(j,j) = S[k]$.
    \item When this condition fails we have $v \leq S[k]$. Applying
      Lemma~\ref{lemma:top} we obtain $S[k] \leq v$ and therefore conclude
      that $S[k] = v$. Hence RMQ$(j,j) = S[k]$.
    \end{itemize}
  \item The \texttt{Value} command. This command essentially adds a new
    value $v$ at the end of $A$, i.e., it sets $A[j] = v$, where $j$ is now
    the last position of $A$. This implies that $j$ is not yet a marked
    position. Therefore for this command we do not need to consider $i =
    j$ because $j$ is not a member of a position set $P[k']$.

    Thus we only need to consider cases when $i < j$. Consider such an
    index $i$, which moreover belongs to the position set $P[k']$, i.e.,
    $i \in P[k']$. The position $i$ must necessarily occur in some set
    $P_b[k'_b]$, which is a set of positions that exists before the
    \texttt{Value} operation alters the stack. In this case we have by
    induction hypothesis that RMQ$(i,j-1) = S_b[k'_b]$. We now divide the
    proof into two cases:
    \begin{itemize}
    \item When $S_b[k'_b] \leq v$, in which case RMQ$(i,j) = S_b[k'_b]$. In
      this case we only need to show that the \texttt{Value} command does
      not alter the index $k'_b$ of the stack, i.e., that $i \in P[k'_b]$ and
      that $S_b[k'_b] = S[k'_b]$. Therefore the desired property holds for
      $k' = k'_b$. This is imediate as the case hypothesis means that even if
      the \texttt{Value} operation happens to extrude level $k'_b$ to the
      top of the stack it does eliminate it, because
      Lemma~\ref{lemma:sortS1} implies that $S_b[k'_b-1] < S_b[k'_b] \leq v$,
      and therefore the while guard in line~\ref{line:wguard}
      fails.
    \item When $v < S_b[k'_b]$, in which case RMQ$(i,j) = v$. In this case
      the value $S_b[k'_b]$ will be discarded by the \texttt{Value}
      command. Let $k$ correspond to the level that is at the top of the
      stack, after the command. By Lemma~\ref{lemma:top} we have that
      $S[k] \leq v$ combining both these inequalities yields
      $S[k] < S_b[k'_b]$. Using Lemma~\ref{lemma:sortS1} we have that
      $S[k-1] < S[k]$, note that Lemma~\ref{lemma:stwo} guarantees that the
      level $k-1$ exists. Moreover because $k$ is the top level of $S$
      after the command we have $S_b[k-1] = S[k-1]$. Combining these
      relations we obtain that $S_b[k-1] < S_b[k'_b]$, to which we apply
      Lemma~\ref{lemma:sortS2}, to conclude that $k-1 < k'_b$. Therefore
      either $k = k'_b$ or the level $k'_b$ was excluded from the stack. In
      both cases position $i$ must be in $P[k]$, either because it was
      already there or it was eventually transferred by the union commands
      in line~\ref{line:union}. Hence we only need to check that
      $S[k] = v$. Let $k_b$ be the \texttt{Top} of stack $S_b$ before the
      command is executed. Hence $k'_b \leq k_b$ and by
      Lemma~\ref{lemma:sortS1} we obtain $S_b[k'_b] \leq S_b[k_b]$. Using
      this case hypothesis and transitivity we obtain that $v <
      S_b[k_b]$. This implies that the condition of the \textbf{if} in
      line~\ref{line:wif} of Algorithm~\ref{alg:value} is true. Therefore
      line~\ref{line:setTop} eventually executes and obtains the condition
      $S[k]= v$ as desired.
    \end{itemize}
  \end{itemize}

\end{proof}
\subsection{Analysis}
\label{sec:analysis}
In this section we discuss several issues related to the performance of our
data structure. Namely we start off by reducing the space requirements from
$O(m)$ to $O(\ell)$. First we need to notice in which ways our data
structure can waist space. In particular the \texttt{Close} command waists
space in the stack itself. In the rightmost structure of
Figure~\ref{fig:stack1} we have that the set $P[3]$ becomes empty after the
\texttt{C 3} command. This set which corresponds to $S[3] = 26$ on the
stack. In essence the item $S[3]$ is no longer necessary in the
stack. However it is kept inactive in the stack, the hash table and the UF
data structure. It is marked as inactive in the hash table, but it still
occupies memory.

Recall that our data structure consists of three components: a stack, a
hash table and a Union-Find data structure. These structures are linked as
follows: the stack contains values and pointers to the hash table; the
hash-table uses sequence positions as keys and UF elements as values; the
Union-Find data structure is used to manipulate sets of reduced positions
and each set in turn points back to a stack position.

Let us now use an amortizing technique to bound the space requirements of
this structure. We start off by allocating a data structure that can
contain at most $a$ elements, where $a$ is a small initial
constant. Allocating a structure with this value implies the following
guarantees:
\begin{itemize}
\item It is possible to insert $a$ elements into the stack without overflow.
\item It is possible to insert $a$ elements into the hash table and the
  overall occupation is always less than half. This guarantees average and
  high probability efficient insertions and searches.
\item It is possible to use $a$ positions for Union-Find operations.
\end{itemize}
Hence we can use this data structure until we reach the limit $a$. When the
limit is reached we consider the number of currently active marked
positions, i.e., the number of positions $i$ such that \texttt{M} was
issued at position $i$, but up to the current position no \texttt{Close}
$i$ was never issued. To determine this value it is best to keep a counter
$c$. This counter is increased when a \texttt{Mark} command is issued,
unless the previous command was also a \text{Mark} command, in which case
it is a repeated marking for a certain position. The counter is decreased
when a \texttt{Close} $i$ is issued, provided position $i$ is currently
active, i.e., it was activated by some \texttt{Mark} command and it has not
yet been closed by any other \texttt{Close} command. Hence by consulting
this counter $c$ we can determine in $O(1)$ time the number of active
positions at this instant. We can now alloc a new data structure with
$a' = 2c$, i.e., a data structure that can support twice as many elements
as the number of current active positions. Then we transfer all the active
elements from the old data structure to the new data structure. The process
is fairly involved, but in essence it requires $O(a \times \alpha(a))$ time
and when it finishes the new data structure contains all the active
positions, which occupy exactly half of the new data structure. This factor
is crucial as it implies that the amortized time of this transfer is in
fact $O(\alpha(a))$ and moreover that the allocated size is at most
$O(2\ell)$.

We now describe how to transfer only the active elements from the old data
structure to the new data structure. First we mark all the elements in the
old stack as inactive. In our implementation we make all the values
negative, as the test input sequences contained no negative values but
other marking schemes may be used. This is also the scheme we used to mark
inactive hash entries.

Now traverse the old hash table and copy all the active values to the new
hash table. Also initilize the pointers from the new hash table to the new
UF data structure. The new UF positions are initialized incrementally,
starting at $1$. Hence every insertion into the new hash function creates a
new UF position, that is obtained incrementally from the last one. We also
look up the old UF positions that are given by active entries of the old
hash table. We use those old active sets to reactivate the old stack
entries. This process allowed us to identify which stack entries are
actually relevant in the old stack. With this information we can compact
the old stack by removing the inactive positions. We compact the old stack
directly to the new stack, so the new stack contains only active
positions. We also add pointers from the old stack to the new stack. Each
active entry of the old stack points to its correspondent in the new
stack. In our implementation this was done by overriding the pointers to
the old hash table, as they are no longer necessary.

At this point the new stack contains the active values, but it still has
not initialized the pointers to the new hash table. These pointers are in
fact position values, because positions are used as keys in the
hash-table. To initialize these pointers we again traverse the active
entries of the old hash table and map them to the old UF positions and to
the corresponding old stack items. We now use the pointer from the old
stack item to the new stack item and update the position pointer of the new
stack to the key of the active entry of the new hash that we are
processing. This assignment works because positions are kept invariant
from the old data structure to the new one. Therefore these positions are
also keys of the new hash. We finish this process by updating the pointers
of the new UF data structure to point to the corresponding items of the new
stack. Since we now know the active items in the new stack and have pointers
from the new stack to the new hash and from the new hash to the new UF
position, we can simply assign the link from the new UF set back to the
item of the new stack item. Thus closing this reference loop.

At this point almost all of the data structure is linked up. The new stack
points to the new hash table, the new hash table points to the new UF
structure and the sets of the new UF structure point to the new stack. The
only missing ingredient is that the sets of the new UF structure are still
singletons, because no \texttt{Union} operations have yet been issued. The main
observation to recover this information is that several positions in the
new UF structure point to the same item in the new stack. Those positions
need to be united into the same set. To establish these unions we traverse
the new UF data structure. For each UF position we determine its
corresponding stack item, note that this requires a Find operation. We then
follow its pointer to an item in the new hash, and a pointer from that item
back to a position in the new UF data structure. Now we unite two UF sets,
the one that contained the initial position and the one that contains the
position that was obtained by passing through the stack and the hash.

\begin{theorem}
  It is possible to process online a sequence of RMQ commands in $O(\ell)$
  space using $O(\alpha(\ell))$ expected amortized time per command.
  \label{teo:basic}
\end{theorem}
\begin{proof}
  The discussion in this section essentially establishes this result. We
  only need to point out the complexities of the data structures that we
  are using. As mentioned before the UF structure requires $O(\alpha(n))$
  amortized time. The stack is implemented over an array and therefore
  requires $O(1)$ per \texttt{Push} and \texttt{Pop} command. In theory we
  consider a hash-table with separate chaining and amaximum load factor of
  50\%, which obtains $O(1)$ expected time per operation. In practice we
  implemented a linear probing approach.

  The final argument is to show that the transfer process requires
  $O(\alpha(\ell))$ amortized time. Whenever a transfer process terminates
  the resulting structure is exactly half full. As the algorithm progresses
  elements are inserted into the structure until it becomes full. Whenever
  an element is inserted we store $2$ credits. Hence when the structure is
  full there is a credit for each element it contains, therefore there are
  enough credits to amortize a full transfer process. We assume that these
  credits are actually multiplied by $\alpha(\ell)$ and whatever is the
  constant of the transfer procedure is.
\end{proof}
One important variation of the above procedure is the offline version of
the problem. Meaning that we are given the complete sequence of commands
and are allowed to process them as necessary to obtain better
performance. In this case we can use a more efficient variant of the Union
Find data structure and obtain $O(1)$ time per
operation, proposed by~\citet*{Gabow_1985}.
\begin{corolary}
  It is possible to process offline a sequence of RMQ commands in $O(\ell)$
  space using $O(1)$ expected amortized time per command.
  \label{cor:const}
\end{corolary}
On the other extreme of applications we may be interrested in real time
applications. Meaning that we need to focus on minimizing the worst case
time that is necessary to process a given command. In this case we can
modify our data structure to avoid excessively long operations, i.e.,
obtain stricter bounds for the worst case time. As an initial result let us
de-amortize the transfer procedure, assuming the same conditions as in
Theorem~\ref{teo:basic}.
\begin{lemma}
  Given a sequence of RMQ commands it is possible to processes them so that
  the transfer procedures require an overhead of $O(\alpha(\ell))$ expected
  amortized time per command.
  \label{lemma:transfer}
\end{lemma}
\begin{proof}
  Note that the transfer process requires $O(a\times \alpha(a))$ amortized
  time to transfer a structure that supports $a$ elements.

  We modify the transference procedure so that it transfers two full
  structures at the same time, by merging their active elements into a new
  structure. The process is essentially similar to the previous
  transference procedure, with a few key differences.

  An element can only be considered active if it is not marked as inactive
  in one of the old hashes. More precisely: if it is marked as active in
  one hash and as inactive in the other then it is inactive; if it is
  marked as active in one hash and does not exists in the other then it is
  active; if it is marked as active in both then it is active.

  Once the active elements of the old stacks are identified they are merged
  into the new stack, by using the same merging procedure that is used in
  mergeSort algorithm, with the proviso that there should be only one copy
  of the sentinel in the merged stack. The third important sincronization
  point is the union commands. Before starting this process it is necessary
  that all the information from the old structures has been transfered to
  the new one, recall that this process generaly iterates over the new
  structure, not the old ones.

  When the old structures can support $a_1$ and $a_2$ elements respectively
  the merging process requires $O(a_1+a_2)$ operations. Note that we do not
  mean time, instead we mean primitive operations on the data structures
  that compose the overall structure, namely accessing the hash function,
  following pointers or calling union or find. Given this merging primitive
  we can now deamortize our transfer process. Instead of immediately
  discarding a structure that hits its full occupancy we keep it around
  because we can not afford to do an immediate transfer. Instead when we
  have at least two full structures we initiate the transfer process. Again
  to avoid exceeding real time requirements this process is kept running in
  parallel, or interleaved, with the processing of the remaining commands
  in the sequence. Since this procedure requires $O(a_1+a_2)$ operations,
  it is possible to tune it to guarantee that it is terminated by the time
  that at most $(a_1 + a_2)/2$ commands are processed. In this case each
  command only needs to contribute $O(1)$ operations to the merging
  process. Each operation requires has an expected $O(\alpha(\ell))$ time,
  which yields the claimed value.

  Hence, at any given instant, we can have several structures in memory. In
  fact we can have at most four, which serve the following purporses:
  \begin{itemize}
  \item One active structure. This structure is the only one that is
    currently active, meaning that it is the only structure that still
    supports \texttt{Mark} and \texttt{Value} commands.
  \item Two static full structures that are currently being merged.
  \item One destination structure that will store the result of the merged
    structures. In general this structure is in some inconsistent state and
    does not process \texttt{Query} commands. The only command that it
    accepts is \texttt{Close}.
  \end{itemize}
  At any point of the execution some or all of the previous structures may
  be in memory. The only one that is always guaranteed to exist is the
  active structure. Now let us discuss how to process commands with these
  structures.
  \begin{itemize}
  \item The \texttt{Query} command is processed by all structures, except
    the destination structure which is potentially inconsistent. From the
    three possible values we return the overall minimum. In this case we
    are assuming that if the query position $i$ is smaller than the minimum
    position index stored in the structure than it returns its minimum
    value, i.e., the value above the $-\infty$ sentinel.
  \item The \texttt{Mark} and \texttt{Value} commands modify only the
    active structure.
  \item The \texttt{Close} command is applied to all the structures,
    including the destination structure. This causes no conflict or
    inconsistency. Recall that elements are not removed from the
    hashes, they are only marked as inactive.
  \end{itemize}

  If we have only the active structure in memory, we use it to process the
  \texttt{Mark} and \texttt{Value} commands. When this active structure
  gets full we mark it as static and ask for a new structure that supports
  the same number $a$ of elements. This structure becomes the new active
  structure. Note that requesting memory may require $O(a)$ time, assuming
  we need to clean it. This can be mitigated by using approaches
  as~\citet*{Briggs_1993} or assuming that this process was
  previously executed, which is possible with in our approach.

  As soon as the second structure becomes full we start the merging process
  to a new destination structure. We consult the number of active elements
  in each one, $c_1$ and $c_2$. We request the destination structure to
  support exactly $c_1+c_2$ elements. This implies that once the merge
  procedure is over the destination structure is full and no further
  elements can be inserted into it. At which point we need to request
  another active structure. If the full structures have sizes $a_1$ and
  $a_2$ we ask for an active structure that can support $(a_1+a_2)/2$
  elements. As argued above this active structure only gets full after the
  merging process finishes. At that point the original full structures can
  be discarded and again we have two full structures, the result of the
  previous merger and the filled up active structure. At this point we
  repeat the process.

  The reason to have a division by $2$ associated with $a_1+a_2$ is that
  its iteration yields a geometric series that does not exceed $2
  \ell$. Hence implying that none of the structures need to support more
  that $2\ell$ elements. This can also be verified by induction. Assuming
  that the original alloc size $a$ is also less than $2\ell$, we have by
  induction hypothesis that $a_1 \leq 2 \ell$ and $a_2 \leq 2\ell$
  therefore $(a_1+a_2)/2 \leq (2 \ell+2 \ell)/2 \leq 2 \ell$. Also by the
  definition of $\ell$ we also have that $c_1 \leq \ell$ and
  $c_2 \leq \ell$ which implies that the destination structures also
  support at most $2\ell$ elements. Since the algorithm uses a at most $4$
  structures simultaneously, we can thus conclude that the overall space
  requirements of the procedure is $O(\ell)$.
\end{proof}
Note that in the worst case the time bound of the UF structures is
$O(\log \ell)$ rather than $O(\alpha(\ell))$. Also note that using a strict
worst case analysis would yield an $O(\ell)$ worst case time for our
complete data structure. Because it contains a hash-table. To avoid this
pathological analysis we instead consider a high probability upper
bound. In this context we obtain an $O(\log \ell)$ time bound with high
probability, for all commands except the \texttt{Value} command. Hence let
us now address this command.
\begin{theorem}
  It is possible to process, in real time, a sequence of RMQ commands in
  $O(\ell)$ space and in $O(\log \ell)$ time per operation with high
  probability.
  \label{teo:high}
\end{theorem}
\begin{proof}
  Given the previous observations we can account $O(\log \ell)$ time for
  the UF structure and the hash table, with high probability,
  see~\citet{mitzenmacher2017probability}.  Lemma~\ref{lemma:transfer}
  de-amortized the transfer operation, hence in this proof we only need to
  explain how to de-amortize the \texttt{Value} operation.

  Algorithm~\ref{alg:value} specifies that given an argument $v$ this
  procedure removes from the stack $S$ the elements that are strictly
  larger than $v$. This process may end up removing all the elements from
  the stack, except obviously the $-\infty$ sentinel. Hence its worst case
  time is $O(m)$, where $m$ is the maximum number of elements in the
  stack. The transfer procedure guarantees that the stack does not
  accumulate deactivated items and therefore we have that $m =
  O(\ell)$. This is still too much time for a real time operation. Instead
  we can replace this procedure by a binary search over $S$, i.e., we
  assume that stack is implemented on an array and therefore we have direct
  access to its elements in constant time. As shown in
  Lemma~\ref{lemma:sortS1} the elements of $S$ are sorted. Therefore we can
  compute a binary search for the position of $v$ and discard all the
  elements in $S$ that are larger than $v$ in $O(\log \ell)$ time. Recall
  that we use variable $k$ to indicate the top of the stack. Once the
  necessary position is identified we update $k$.

  However Algorithm~\ref{alg:value} also specifies that each element that
  is removed from the stack invokes a \texttt{Union} operation,
  line~\ref{line:union}. To perform these unions in real time we need a
  different UF data structure.

  Most UF structures work by choosing a representative element for each
  set. The representative is the element that is returned by the
  \texttt{Find} operation. This representative is usually an element of the
  set it represents. The representative either posseses, or is assigned,
  some distinct feature that makes it easy to identify. In the UF structure
  by~\citet*{Tarjan_1984} a representative is stored at the root of a tree.

  Lemma~\ref{lemma:sortP} essentially states that the sets that we are
  interrested in can be sorted, without incosistencies among elements of
  diferent sets. Hence this provides a natural way for choosing a
  representative. Each set can be represented by its minimum element. With
  this representation the $\mathtt{Find}(p)$ operation consists in finding
  the largest representative that is still less than or equal to $p$, i.e.,
  the Predecessor. The Union operation simply discards the largest
  representative and keeps the smallest one. Hence we do not require an
  extra data structure, it is enough to store the minimums along with
  values within the stack items. To compute the Predecessors we perform a
  binary search over the minimums. This process requires $O(\log \ell)$
  time. Moreover the variable $k$ allows us to perform multiple
  \texttt{Union} operations at once. Let us illustrate how to use this data
  structure for our goals. Recall the sample command sequence:

\begin{verbatim}
  V 22 M V 23 M V 26 M V 28 M V 32 M V 27 M V 35 M Q 4 C 3
\end{verbatim}

  Now assume that after this sequence we also execute the command
  \verb|V 10|. We illustrate how a representation based on minimums
  processes these commands, Figure~\ref{fig:stack3}. The structure on left
  is the configuration after the initial sequence of commands. The
  structure in the middle represents the actual configuration that is
  stored in memory. Note that for each set we store only its minimum
  element. In particular note that the set associated with value $26$ is
  represented by $3$, even though position $3$ was already marked as
  closed. As mentioned the hash-table keeps track of which positions are
  still open and closed positions are removed during transfer
  operations. This means that until then it is necessary to use all
  positions, closed or not, for our UF data structure. Hence the
  representative of a set is the minimum over all positions that are
  related to the set, closed or not. The structure on the right represents
  the structure after processing the \verb|V 10| command.

\renewcommand{\sItem}[1]{
  \psframebox[framearc=0.7,fillstyle=solid]{\bhBox{10pt}{24pt}{$#1$}}
}
\def\sbItem#1{
  \psframebox[framearc=0.7,fillstyle=solid,fillcolor=black!50]{\bhBox{10pt}{24pt}{
      \textcolor{white}{$\mathbf{#1}$}}}
}

  \begin{figure}[tb]
  \centering
    \begin{pspicture}(12.0,5)
      \rput[bl](0,0){
        \psframebox*[fillstyle=solid,fillcolor=black!4]{
          \begin{psmatrix}[mnode=r,rowsep=5pt,colsep=9pt,mcol=l]
            \sItem{35} & $\{7\}$ \\
            \sItem{27} & $\{4,5,6\}$ \\
            \sItem{26} & $\emptyset$ \\
            \sItem{23} & $\{2\}$ \\
            \sItem{22} & $\{1\}$ \\
            \sItem{-\infty}
          \end{psmatrix}}}
      \rput[bl](4.0,0){
        \psframebox*[fillstyle=solid,fillcolor=black!4]{
          \begin{psmatrix}[mnode=r,rowsep=5pt,colsep=9pt,mcol=l]
            \sItem{35} & $7$ \\
            \sItem{27} & $4$ \\
            \sItem{26} & $3$ \\
            \sItem{23} & $2$ \\
            \sItem{22} & $1$ \\
            \sItem{-\infty}
          \end{psmatrix}}}
      \rput[bl](8.2,0){
        \psframebox*[fillstyle=solid,fillcolor=black!4]{
          \begin{psmatrix}[mnode=r,rowsep=5pt,colsep=9pt,mcol=l]
            \sItem{10} & $1$ \\
            \sItem{-\infty}
          \end{psmatrix}}}
    \end{pspicture}
    \caption[Stack configuration]{Illustration structure configuration
      using minimums to represent position sets.}
  \label{fig:stack3}
\end{figure}
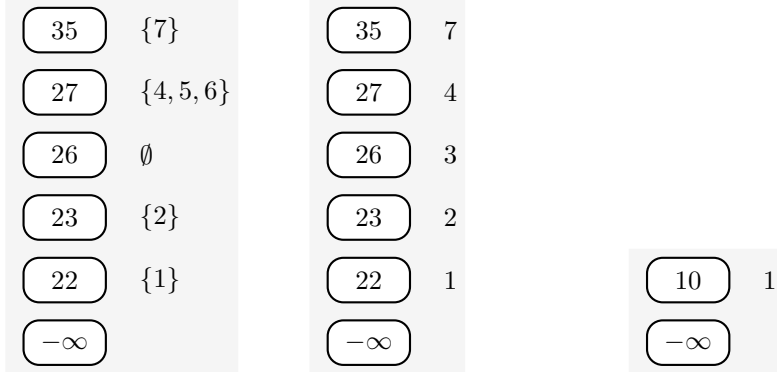
Note that in this final configuration the set, of active positions,
associated with value $10$ should be $\{1,2,4,5,6,7\}$. However it is
represented only by the value $1$. This set should be obtained by the
following sequence of \texttt{Union} operations
$\{1\} \cup \{2\} \cup \{4,5,6\} \cup \{7\}$. This amounts to removing the
numbers $2$, $4$ and $7$, which is obtained automatically when we alter the
variable $k$.

Summing up, our data structure consists of the following elements:
  \begin{itemize}
  \item An array storing stack S. Each element in the stack contains a
    value $v$ and position $i$, which is the minimum of the position set it
    represents.
  \item A hash-table to identify the active positions. In this
    configuration no mapping is required, it is enough to identify the
    active positions.
  \end{itemize}

  The general procedure for executing commands and the respective time
  bounds are the following:
  \begin{itemize}
  \item The \texttt{Value} command needs to truncate the stack, by updating
    variable $k$. This process requires $O(\log \ell)$ time because of the
    binary search procedure, but it can actually be improved to
    $O(1+\log d)$ time where $d$ is the number of positions removed from
    the position tree, by using an exponential search that starts at the
    top of the stack. Using an exponential search the expected
    amortized time of this operation is $O(1)$.
  \item The \texttt{Mark} command needs to add an element to the hash-table
    and an element to the stack $S$. This requires $O(\log \ell)$ time with
    high probability. The \texttt{Make-Set} or \texttt{Union} operations
    require only $O(1)$ time hence the overall time is dominated by
    $O(\log \ell)$. The expected time of this operation is $O(1)$.
  \item The \texttt{Query} command needs to search for an element in the
    hash-table and compute a \texttt{Find} operation. The \texttt{Find}
    operation is computed with a binary search over minimums stored in the
    items of the stack. This operation requires $O(\log \ell)$ time with
    high probability. The expected amortized time is also $O(\log \ell)$,
    but it can be improved to $O(1+\log(j-i+1))$ for a query with indexes
    $(i,j)$, by using an exponential search from the top of the stack.
  \item The \texttt{Close} command needs to remove an element from the
    hash-table. This requires $O(\log \ell)$ time with high probability and
    $O(1)$ expected time.
  \end{itemize}
\end{proof}

The data structure of the previous theorem is simple because most of the
complex de-amortizing procedure is handled in
Lemma~\ref{lemma:transfer}. We now focus on how to further reduce the high
probability time bounds to $O(\log \log n)$. A simple way to obtain this is
to have $\ell = O(\log n)$, i.e., having at most $O(\log n)$ active
positions at each time. This may be achieved if \texttt{Query} positions
are not necessarily exact, meaning that the data structure actually returns
the solution for a query $(i',j)$ instead of $(i,j)$. The goal is that
$j-i$ is similar in size of $j-i'$. Meaning that $j-i \leq j-i' <
2(j-i)$. In this scenario it is enough to keep $O(\log n)$ active
positions, i.e., positions $i'$ for which $j - i' = 2^c$ for some integer
$c$. Since the data structure of Theorem~\ref{teo:high} does not use the
hash-table to reduce the position range, we can bypass its use in these
queries. It is enough to directly determine the predecessor of $i$ among
the minimums stored in the stack $S$. Which is computed with a binary
search or exponential search as explained in the proof.

The problem with this specific set of positions is that when $j$ increases
the active positions no longer provide exact powers of two. This is not
critical because we can adopt an update procedure that provides similar
results. Let $i_1 < i_2 < i_3$ represent three consecutive positions that
are currently active. When $j$ increases we check whether to keep $i_2$ or
discard it. It is kept if $j - i_1 > 2 (j - i_3)$, otherwise it is
discarded. Hence we keep a list of active positions that gets updated by
adding the new position $j$ and checking two triples of active
positions. We keep an index that indicates which triple to check and at
each step use it to check two triples, moving from smaller to larger
position values. The extremes of the list are not checked. We show the
resulting list of positions in Table~\ref{tab:aplists}, where the bold
numbers indicate the triples that will be checked in the next
iteration. Whenever the triples to check reach the end of the list we have
that the size of the list is at most $2 \log_2 n$, because the verification
guarantees that the value $j-i$ is divided in half for every other position
$i$. Therefore it takes at most $2 \log_2 n$ steps to traverse the
list. Hence this list can contain at most $4 \log_2 n = O(\log n)$
positions and each time $j$ is updated only $O(1)$ time is used.
\begin{table}[tb]
  \centering
  \begin{tabular}{*{11}{r}}
\textbf{1} & \\
\textbf{1} & \textbf{2} & \\
\textbf{1} & \textbf{2} & \textbf{3} & \\
\textbf{1} & \textbf{2} & \textbf{3} & 4 & \\
1 & \textbf{3} & \textbf{4} & \textbf{5} & \\
\textbf{1} & \textbf{3} & \textbf{4} & 5 & 6 & \\
1 & \textbf{4} & \textbf{5} & \textbf{6} & 7 & \\
1 & 4 & \textbf{6} & \textbf{7} & \textbf{8} & \\
\textbf{1} & \textbf{4} & \textbf{6} & 7 & 8 & 9 & \\
1 & \textbf{4} & \textbf{7} & \textbf{8} & 9 & 10 & \\
1 & 4 & \textbf{7} & \textbf{9} & \textbf{10} & 11 & \\
1 & 4 & 7 & 9 & \textbf{10} & \textbf{11} & \textbf{12} & \\
\textbf{1} & \textbf{4} & \textbf{7} & 9 & 10 & 11 & 12 & 13 & \\
1 & \textbf{7} & \textbf{9} & \textbf{10} & 11 & 12 & 13 & 14 & \\
1 & \textbf{7} & \textbf{11} & \textbf{12} & 13 & 14 & 15 & \\
1 & 7 & \textbf{11} & \textbf{13} & \textbf{14} & 15 & 16 & \\
1 & 7 & 11 & \textbf{14} & \textbf{15} & \textbf{16} & 17 & \\
1 & 7 & 11 & 14 & \textbf{16} & \textbf{17} & \textbf{18} & \\
\textbf{1} & \textbf{7} & \textbf{11} & 14 & 16 & 17 & 18 & 19 & \\
1 & 7 & \textbf{11} & \textbf{14} & \textbf{16} & 17 & 18 & 19 & 20 & \\
1 & 7 & 11 & \textbf{16} & \textbf{17} & \textbf{18} & 19 & 20 & 21 & \\
1 & 7 & 11 & \textbf{16} & \textbf{19} & \textbf{20} & 21 & 22 & \\
1 & 7 & 11 & 16 & \textbf{19} & \textbf{21} & \textbf{22} & 23 & \\
1 & 7 & 11 & 16 & 19 & 21 & \textbf{22} & \textbf{23} & \textbf{24} & \\
\textbf{1} & \textbf{7} & \textbf{11} & 16 & 19 & 21 & 22 & 23 & 24 & 25 & \\
1 & \textbf{11} & \textbf{16} & \textbf{19} & 21 & 22 & 23 & 24 & 25 & 26 & \\
1 & 11 & \textbf{19} & \textbf{21} & \textbf{22} & 23 & 24 & 25 & 26 & 27 & \\
1 & 11 & \textbf{19} & \textbf{23} & \textbf{24} & 25 & 26 & 27 & 28 & \\
1 & 11 & 19 & \textbf{24} & \textbf{25} & \textbf{26} & 27 & 28 & 29 & \\
1 & 11 & 19 & \textbf{24} & \textbf{27} & \textbf{28} & 29 & 30 & \\
  \end{tabular}
  \caption{Sequence of active position lists}
  \label{tab:aplists}
\end{table}

Another alternative for obtaining $O(\log \log n)$ high probability time is
to change the UF structure. In this case we use the same approach as
Theorem~\ref{teo:high} that relies on predecessor searches to compute the
\texttt{Find} operation. This time we consider the Van Emde Boas tree that
supports this operation efficiently, but requires longer to update.

\begin{theorem}
  It is possible to process, in real time, a sequence of RMQ commands in
  $O(\ell)$ space and in $O(\log \log \ell)$ time with high
  probability, for all operations except \texttt{Value}, which requires
  $O(\sqrt{\ell})$ time with high probability.
  \label{teo:real}
\end{theorem}
\begin{proof}
  First note that the \texttt{Value} command is not used in the
  de-amortized transfer procedure described in
  Lemma~\ref{lemma:transfer}. Thus guaranteeing that the overhead per
  command will be only $O(\log \log \ell)$ time, once the statement of the
  Theorem is established.

  One important consideration is to reduce the high probability time of the
  hash-table to $O(\log \log \ell)$ instead of $O(\log \ell)$. For this
  goal we modify the separate chaining to the 2-way chaining approach
  proposed by~\citet*{Azar_1999}, also with a maximum load factor of 50\%.

  We can now analyze the Van Emde Boas tree (VEB). This data structure is
  used as in Theorem~\ref{teo:high} to store the minimum values of each
  set. Hence the underlying universe are the positions over $A$. Since this
  structure uses linear space in the universe size this would yield $O(n)$
  space. However in this case we can use the hash-table to reduce the
  position range and thus the required space becomes $O(\ell)$. Note that
  the reduced positions are also integers and we can thus correctly use
  this data structure.

  Given that the time to compute a predecessor with this data structure is
  $O(\log \log \ell)$ this then implies this bound for the RMQ operations
  except \texttt{Value}. For this operation we have two caveats. First the
  binary search over the values in the stack $S$ still requires
  $O(\log \ell)$ time. Second the \texttt{Union} operations in
  Algorithm~\ref{alg:value} implies that it is necessary to remove elements
  from the VEB tree. This is not a problem for the \texttt{Mark} operation,
  Algorithm~\ref{alg:Mark}, because a single removal in this tree also
  requires $O(\log \log \ell)$ time. The issue for \texttt{Value} is that it
  may perform several such operations. In particular when $d$ elements are
  removed from the stack it requires $O(d \log \log \ell)$ time. Recall the
  example in the proof of Theorem~\ref{teo:high}, where several union
  operations where executed to produce the set
  $\{1\} \cup \{2\} \cup \emptyset \cup \{4,5,6\} \cup \{7\}$. In that
  Theorem this was done automatically by modifying $k$, but in this case it
  is necessary to actually remove the elements $2$, $3$, $4$ and $7$ from
  the VEB tree. Note that the element $3$ is the representative of the
  empty set. Even though it is not active it was still in the VEB tree.

  This consists in removing from the VEB tree all the elements that are
  larger than $1$. The VEB tree does not have a native operation for this
  process. Hence we have thus far assumed that this was obtained by
  iterating the delete operation. Still it is possible to implement this
  bulk delete operation directly within the structure, much like it can be
  done over binary search trees. In essence the procedure is to directly
  mark the necessary first level structures as empty and then do a double
  recursion, which is usually strictly avoided in this data
  structure. Given a variable $u$ that identifies the logarithm of the
  universe size as $\ell = 2^u$, this yields the following time recursion
  $T(u) = 2^{u/2} + 2 T(u/2)$. Note that $2^{u/2} = \sqrt{\ell}$ is the
  number of structures that exist in the first level, and potentially need
  to be modified. This recursion is bounded by
  $O(2^{u/2}) = O(\sqrt{\ell})$.
\end{proof}

As a final remark about this last result note that the time bound for the
\texttt{Value} command is also $O(\log \log \ell)$ amortized, only the high
probability bound is $O(\sqrt{\ell})$. This is because the iterated
deletion bound $O(d \log \log \ell)$ that we mentioned in the proof does
amortize to $O(\log \log \ell)$ and for each instance of the \texttt{Value}
command we can choose between $O(d \log \log \ell)$ and $O(\sqrt{\ell})$.

This closes the theoretical analysis of the data structure. Further
discussion is given in Section~\ref{sec:disc-concl}.
\section{Experimental}
\label{sec:experimental}
Let us now focus on testing the performance of this structure
experimentally. We implemented the data structure that is described in
Theorem~\ref{teo:basic}. We also designed a generator that produces random
sequences of RMQ commands. In these generated sequences the array $A$
contained $2^{28}$ integers, i.e., $n = 2^{28}$. Each integer was chosen
uniformly between $0$ and $2^{30}-1$, with the \texttt{arc4random\_uniform}
function\footnote{\url{https://github.com/freedesktop/libbsd}}.

We first implemented the version of our Algorithm described in
Section~\ref{sec:solution}, i.e., without using a hash table nor the
transfer process. We refer to the prototype as the vanilla version and use
the letter V to refer to it in our tables. We also implemented the version
described in Theorem~\ref{teo:basic}, which includes a hash table and
requires a transfer process. We use the label T2 to refer to this
prototype.

For a baseline comparison we used the \texttt{ST-RMQ-CON} algorithm
by~\citet*{Alzamel_2018}. We obtained the implementation from their github
repository \url{https://github.com/solonas13/rmqo}.

Our RMQ command sequence generator proceeds as follows. First it generates
$n = 2^{28}$ integers uniformly between $0$ and $2^{30}-1$. Then it chooses
a position to \texttt{Mark}, uniformly among the $n$ positions
available. This process is repeated $q$ times. Note that the choices are
made with repetition, therefore the same position can be chosen several
times. Each marked position in turn will force a query command. All query
intervals have the same length $l = j-i+1$. Under these conditions it is
easy to verify that the expected number of open positions at a given time
is $l\times q/n$ and the actual number should be highly concentrated around
this value. Hence we assume that this value corresponds to our $\ell$
parameter and therefore determine $l$ as $\ell \times n /q$.

The tests were performed on a $64$ bit machine, running \texttt{Linux mem
  4.19.0-12}, which contained 32 cores in \texttt{Intel(R) Xeon(R) CPU E7-
  4830 @ 2.13GHz} CPUs. The system has $256 Gb$ of RAM and of swap. Our
prototypes were compile with \texttt{gcc 8.3.0} and the baseline prototype
with \texttt{g++}. All prototypes are compiled with \texttt{-O3}. We
measure the average execution time by command and the peak memory used by
the prototypes. These values were both obtained with the system
\texttt{time} command. These results are show in table~\ref{tab:time}
and~\ref{tab:space}.
\begin{table}
  \begin{adjustbox}{max width=\textwidth}
    \begin{tabular}{l|l|*{17}{r}l}
      \input{timetable}
    \end{tabular}
  \end{adjustbox}
  \caption{Execution time per command in nano seconds. The values are
    obtained by dividing total execution time by $n+q$.}
  \label{tab:time}
\end{table}
\begin{table}
  \begin{adjustbox}{max width=\textwidth}
    \begin{tabular}{l|l|*{17}{r}l}
      \input{spacetable}
    \end{tabular}
  \end{adjustbox}
  \caption{Total memory peak in Megabytes, or in Gygabytes when indicated
    by Gb.}
  \label{tab:space}
\end{table}

The results show that our prototypes are very efficient. In terms of time
both V and T2 obtain similar results, see Table~\ref{tab:time}. As expected
T2 is slightly slower than V, but in practice this different is less than a
factor of $2$. The time performance of B is also very similar, in fact V
and T2 are faster, which was not expected as B has $O(1)$ performance per
operation and V and T2 have $O(\alpha(n))$. Even though in practice this
difference was expected to be very small we were not expecting to obtain
faster performance. This is possibly a consequence of the memory hierarchy
as B works by keeping $A$ and all the queries in memory.

Concerning memory our prototypes also obtained very good performance, see
Table~\ref{tab:space}. In particular we can clearly show a significant
difference between using $O(q)$ and $O(\ell)$ extra performance. Consider
for example $q = 2^{26}$ and $\ell = 2^{16}$. For these values V uses more
than one gigabyte of memory, whereas T2 requires only 17Mb, a very large
difference. In general T2 uses less memory than V, except when $q$ and
$\ell$ become similar. For example when $q = \ell = 2^{26}$ V use around
one gigabyte of memory, whereas T2 requires three, but this is expected. Up
to a given fixed factor. The baseline B requires much more memory as it
stores more items in memory. Namely a compacted version of the array $A$
and the solutions to all of the queries. Our prototypes V and T2 do not
store query solutions. Instead whenever a query is computed its value is
written to a \texttt{volatile} variable. This guarantees that all the
necessary computation is performed, instead of optimized away by the
compiler. However it also means that previous solutions are overwritten by
newer results. We deemed this solution as adequate for an online algorithm,
which in practice will most likely pass its results to a calling
process. Moreover storing the query solutions would bound the
experimental results to $\Omega(q)$ space, thus not being a fair test of
$O(\ell)$ space.
\section{Related Work}
\label{sec:related-work}
The Range Minimum Query problem has been exhaustively studied. This problem
was shown to be linearly equivalent to the Lowest Common Ancestor problem
in a static tree by~\citet*{Gabow_1984}. A recent perspective on this
result was given by~\citet*{Bender_2000}. The first major solution to the
LCA problem, by~\citet*{Berkman_1993}, obtained $O(\alpha(n))$ time, using
Union-Find data structures. Similarly to our data structure. In fact this
initial result was a fundamental inspiration for the data structure we
propose in this paper. A constant time solution was proposed
by~\citet*{Harel_1984}. A simplified algorithm was proposed
by~\citet*{Schieber_1988}. A simplified exposition of these algorithms, and
linear equivalence reductions, was given by~\citet*{Bender_2000}.

Even though these algorithms were simpler to understand and implement they
still required $O(n)$ space to store auxiliary data structures, such as
Cartesian trees. Moreover the constants associated with these data
structures were large, limiting the practical application of these
algorithms. To improve this limitation direct optimal direct algorithms for
RMQ were proposed by~\citet*{Fischer_2006}. The authors also showed that
their proposal improved previous results by a factor of two. However they
also observed that for several common problem sizes, asymptotically slower
variants obtained better performance. Hence a practical approach, that
obtained a 5 time speedup, was proposed by~\citet*{Ilie_2010}. Their approach
was geared towards the Longest Common Extension on strings and leveraged
the use its average value to.

A line of research directed by an approach that focused on reducing
constants by using succinct and compressed representations was initiated
by~\citet*{Sadakane_2007B} and successively improved
by~\citet*{Sadakane_2007A},~\citet*{Sadakane_2010} and~\citet*{Fischer_2011}. The
last authors provide a systematic comparison of the different results up
to~\citeyear{Fischer_2011}. Their solution provided an $2n+o(n)$ bits data
structure the answers queries in $O(1)$ time.

Still several engineering techniques can be used obtain more practical
efficient solutions. An initial technique was proposed
by~\citet*{Grossi_2013}. A simplification implemented by~\citet*{Ferrada_2017}
used $2.1n$ bits and answered queries in $1$ to $3$ microseconds per
query. Another proposal by~\citet*{baumstark_et_al:LIPIcs:2017:7615} obtained
around a 1 microsecond per query (timings vary depending on query
parameters) on an single core of the Intel Xeon E5-4640 CPU.

A new approach was proposed by~\citet*{Alzamel_2018} where no index data
structure was created by a preprocessing step. Instead all the RMQs are
batched together and solved in $n + O(q)$ time and $O(q)$ space. This space
was used to store a contracted version of the input array $A$ and the
solutions to the queries. This is essentially the approach we follow in
this paper. Therefore in Table~\ref{tab:time} we independently verify their
query times in the nanoseconds. Also table~\ref{tab:space} reports the
memory requirements of their structure. In a recent
result~\citet*{10.1002/spe.2597} proposed an heuristic
idea, without constant worst case time and a hybrid variation with $O(1)$
time and $3n$ bits. Their best result obtains competitive results against
existing solutions, except possibly for small queries. Their results show
query times essentially equal to ours and the algorithm
of~\citet*{Alzamel_2018} for large queries, but they also obtain 10 times
slower performance for small queries.

For completion we also include references to the data structures we used,
or mentioned, in our approach.

The technique by~\citet*{Briggs_1993} provides a way to use memory without
the need to initialize it. Moreover each time a given memory position needs
to be used for the first time it requires only $O(1)$ time to register this
change. The trade-off with this data structure is that it triples the space
requirements. Since, for now, we do not have an implementation of
Lemma~\ref{lemma:transfer}, the claimed result can use this technique, also
explained by~\citet*{bentley2016programming} and~\citet*{aho1974design}. For
our particular implementation this can be overcome. For the destination
structure is not a problem because we can assume that the whole merge
process includes the time for the initial clean-up, all within
$(a_1+a_2)/2$ as explained in Lemma~\ref{lemma:transfer}. Only the active
structure requires some more forethought. In essence when the merge
processes starts and we start using an active structure that supports
$(a_1+a_2)/2$ elements it is a good time to start cleaning a piece of
memory that supports $(a_1+a_2+c_1+c_2)/2$ elements, as this will be the
number of elements of the future active structure. We will start using this
structure when the current merge finishes. Since this number of elements is
at most $a_1+a_2$ it is possible to finish the clean-up when at most
$(a_1+a_2)/2$ operations have executed, by cleaning two element positions
in each operation.

The Union-Find data structure is a fundamental piece of our solution. The
original proposal to represent disjoint sets that can support the
\texttt{Union} and \texttt{Find} operations was
by~\citet*{Galler_1964}. Their complexity was bounded by $O(\log ^{*}(n))$
amortized time per operation by~\citet*{Hopcroft_1973}. The analysis of the
time bound was later refined to $O(\alpha(n))$
by~\citet*{Tarjan_1984}. Lower bound analysis guarantees that these bounds
are optimal~\citet*{Tarjan_1979} and~\citet*{Fredman_1989}. However in the
case where the sequence of operations is known \textit{a priori} it is
possible to obtain $O(1)$ amortized time per operation, as shown
by~\citet*{Gabow_1985}. An exhaustive survey was given
by~\citet*{Galil_1991}. An elementary description of this data structure was
provided by~\citet*{cormen2009introduction}
and~\citet*{sedgewick2011algorithms}.

Hash tables date back to the origin of computers. A history on the subject
and the first theoretical analysis was given
by~\citet*{knuth1963notes}. This analysis established constant expect time
bound. The high probability bound of separate chaining can be derived from
balls and bins model, see~\citet*{mitzenmacher2017probability}. Actually a
better bound was obtained by~\citet*{Gonnet_1981}. The 2-way chaining
hash-table was proposed by~\citet*{Azar_1999}, which also established its
constant expected time and high probability bound.

Exponential searches where proposed by~\citet*{BENTLEY197682}
and~\citet*{Baeza_Yates_2010} and can be used to speed-up the binary search
algorithm when the desired element is close to the beginning or end of a
list. For an introduction to binary search
see~\citet{cormen2009introduction}.

The data structure by~\citet*{Emde_Boas_1976} provides support for
\texttt{Predecessor} queries over integers in $O(\log \log n)$ time, by
recursively dividing a tree along its medium height. For an elementary
description, which requires less space was given
by~\citet*{cormen2009introduction}. The y-fast trie data structure was
proposed by~\citet*{Willard_1983} to reduce the large space requirements of
the Van Emde Boas tree. This data structure obtains the $O(\log \log n)$
time bound, only that amortized. For this reason we did not considered it
in Theorem~\ref{teo:real}. Also in the process the this result describes
x-fast tries.

\section{Discussion and Conclusion}
\label{sec:disc-concl}
We can now discuss our results in context. In this paper we started by
defining a set of commands that can be used to form sequences. Although
these commands are fairly limited they can still be used for several
important applications. First notice that if we are given a list of $(i,j)$
RMQs we can reduce them to the classical context. This can be achieved with
two hash tables. In the first table store the queries indexed by $i$ and on
the second by $j$. We use the first table to issue \texttt{Mark} commands
and the second to issue \texttt{Query} commands. This requires some
overhead but it allows our approach to be used to solve classical RMQ
problems. In particular it will significantly increase the memory
requirements, as occurs in Table~\ref{tab:space} between T2 and B.

Our data structures can be used in online and real-time applications. Note
in particular we can use our commands to maintain the marked positions in a
sliding window fashion. Meaning that at any instant we can issue
\texttt{Query} commands for any of the previous $\ell$ positions. The
extremely small memory requirements of our approach makes our data
structure suitable to be used in routers, switches or in embedded
computation devices with low memory and CPU resources.

The simplest configuration of our data structure consists of a stack
combined with a Union-Find data structure. For this structure we can
formally prove that our procedures correctly compute the desired result,
Theorem~\ref{teo:correct}. We then focused on obtaining the data structure
configuration that yielded the best performance. We started by obtaining
$O(\alpha(n))$ amortized time and $O(q)$ space, see
Theorem~\ref{teo:basic}. This result is in theory slower than the result
by~\citet*{Alzamel_2018}, which obtained $O(1)$ amortized query time. We
compared experimentally these approaches in
Section~\ref{sec:experimental}. The results showed that out approach was
competitive, both in terms of time and space, our prototype V was actually
faster than the prototype B by~\citet{Alzamel_2018}. We also showed that it
was possible for our data structure to obtained $O(1)$ amortized query time
(Corolary~\ref{cor:const}), mostly for theoretical competitiveness. We did
not implement this solution.

We described how to reduce the space requirements down to $O(\ell)$, by
transferring information among structures and discarding structures that
became full, see Lemma~\ref{lemma:transfer}. In theory this obtained the
same $O(\alpha(n))$ amortized time but significantly reduced space
requirements. We also implemented this version of the data structure. In
practice the time penalty was less than a $2$ factor. Moreover, for some
configurations, the memory reduction was considerable, see
Table~\ref{tab:space}.

Lastly we focused on obtaining real time performance. We obtained a high
probability bound of $O(\log n)$ amortized time per query, see
Theorem~\ref{teo:high}. This bound guarantees real time performance. We
then investigated alternatives to reduce this time bound to
$O(\log \log n)$. We proposed two solutions. In one case we considered
approximate queries, thus reducing the necessary amount of active positions
to $O(\log n)$. In the other case we used the Van Emde Boas tree, which
provided a $O(\log \log n)$ high probability time bound for all commands
except \texttt{Value}, see Theorem~\ref{teo:real}. In this later
configuration the \texttt{Value} command actually obtained an
$O(\sqrt{\ell})$ bound, which is large, but the corresponding amortized
value is only $O(\log \log n)$.

\section{Acknowledgements}
\label{sec:acknowledgements}
The work reported in this article was supported by national funds through
Fundação para a Ciência e a Tecnologia (FCT) with reference
UIDB/50021/2020 and project NGPHYLO PTDC/CCI-BIO/29676/2017.
\bibliography{biblio}
\end{document}

%% file: timetable
&& $2^{10}$& $2^{11}$& $2^{12}$& $2^{13}$& $2^{14}$& $2^{15}$& $2^{16}$& $2^{17}$& $2^{18}$& $2^{19}$& $2^{20}$& $2^{21}$& $2^{22}$& $2^{23}$& $2^{24}$& $2^{25}$& $2^{26}$& $\leftarrow \ell$ \\ \hline
$2^{10}$ & T2 & 14 \\
& V & 10 \\
& B & 25 \\ \hline
$2^{11}$ & T2 & 15 & 14 \\
& V & 10 & 10 \\
& B & 20 & 20 \\ \hline
$2^{12}$ & T2 & 15 & 15 & 14 \\
& V & 10 & 10 & 10 \\
& B & 21 & 20 & 20 \\ \hline
$2^{13}$ & T2 & 15 & 14 & 15 & 14 \\
& V & 10 & 10 & 11 & 10 \\
& B & 20 & 20 & 20 & 21 \\ \hline
$2^{14}$ & T2 & 15 & 14 & 15 & 14 & 14 \\
& V & 10 & 10 & 10 & 10 & 10 \\
& B & 21 & 21 & 21 & 20 & 20 \\ \hline
$2^{15}$ & T2 & 15 & 16 & 16 & 16 & 16 & 16 \\
& V & 11 & 11 & 10 & 10 & 10 & 10 \\
& B & 21 & 27 & 27 & 27 & 27 & 24 \\ \hline
$2^{16}$ & T2 & 15 & 15 & 14 & 14 & 15 & 15 & 15 \\
& V & 10 & 10 & 10 & 10 & 10 & 10 & 10 \\
& B & 26 & 26 & 27 & 26 & 26 & 26 & 25 \\ \hline
$2^{17}$ & T2 & 23 & 25 & 14 & 14 & 15 & 14 & 19 & 15 \\
& V & 10 & 10 & 10 & 10 & 10 & 10 & 10 & 10 \\
& B & 26 & 27 & 27 & 26 & 26 & 25 & 27 & 26 \\ \hline
$2^{18}$ & T2 & 15 & 16 & 15 & 14 & 15 & 16 & 15 & 14 & 15 \\
& V & 11 & 10 & 11 & 11 & 11 & 10 & 10 & 11 & 10 \\
& B & 28 & 25 & 28 & 27 & 27 & 28 & 27 & 27 & 26 \\ \hline
$2^{19}$ & T2 & 15 & 16 & 17 & 16 & 15 & 16 & 16 & 17 & 16 & 17 \\
& V & 10 & 10 & 11 & 10 & 11 & 10 & 11 & 11 & 11 & 11 \\
& B & 29 & 30 & 26 & 30 & 29 & 29 & 30 & 29 & 28 & 25 \\ \hline
$2^{20}$ & T2 & 17 & 18 & 17 & 17 & 18 & 17 & 18 & 19 & 18 & 20 & 19 \\
& V & 13 & 12 & 11 & 11 & 12 & 12 & 12 & 13 & 13 & 12 & 12 \\
& B & 32 & 33 & 33 & 33 & 33 & 33 & 33 & 32 & 33 & 28 & 31 \\ \hline
$2^{21}$ & T2 & 19 & 19 & 19 & 19 & 19 & 29 & 20 & 21 & 23 & 22 & 24 & 23 \\
& V & 14 & 14 & 13 & 14 & 14 & 15 & 14 & 15 & 14 & 14 & 15 & 14 \\
& B & 38 & 38 & 38 & 37 & 38 & 39 & 39 & 39 & 40 & 38 & 38 & 35 \\ \hline
$2^{22}$ & T2 & 24 & 24 & 24 & 24 & 25 & 27 & 25 & 27 & 31 & 30 & 32 & 33 & 33 \\
& V & 18 & 17 & 17 & 17 & 17 & 20 & 19 & 19 & 20 & 20 & 19 & 20 & 22 \\
& B & 49 & 49 & 50 & 50 & 50 & 50 & 52 & 51 & 50 & 52 & 51 & 49 & 44 \\ \hline
$2^{23}$ & T2 & 33 & 34 & 33 & 37 & 35 & 40 & 37 & 38 & 42 & 53 & 48 & 50 & 52 & 51 \\
& V & 24 & 25 & 24 & 28 & 27 & 27 & 31 & 31 & 32 & 33 & 32 & 32 & 30 & 35 \\
& B & 65 & 66 & 69 & 71 & 67 & 76 & 78 & 86 & 83 & 89 & 79 & 75 & 74 & 64 \\ \hline
$2^{24}$ & T2 & 51 & 49 & 48 & 50 & 50 & 51 & 52 & 57 & 59 & 67 & 75 & 83 & 85 & 91 & 86 \\
& V & 41 & 41 & 40 & 41 & 37 & 38 & 40 & 45 & 45 & 47 & 46 & 45 & 47 & 49 & 54 \\
& B & 106 & 107 & 115 & 107 & 107 & 114 & 115 & 132 & 126 & 127 & 130 & 135 & 134 & 119 & 114 \\ \hline
$2^{25}$ & T2 & 71 & 72 & 73 & 75 & 75 & 79 & 82 & 82 & 91 & 110 & 124 & 143 & 153 & 162 & 159 & 160 \\
& V & 58 & 59 & 59 & 60 & 60 & 61 & 60 & 71 & 73 & 75 & 75 & 74 & 78 & 83 & 91 & 102 \\
& B & 157 & 155 & 158 & 161 & 156 & 167 & 169 & 172 & 179 & 183 & 187 & 191 & 191 & 186 & 182 & 150 \\ \hline
$2^{26}$ & T2 & 104 & 110 & 106 & 111 & 108 & 116 & 122 & 122 & 139 & 167 & 195 & 224 & 241 & 267 & 266 & 264 & 280 \\
& V & 93 & 94 & 94 & 94 & 95 & 94 & 95 & 111 & 119 & 119 & 120 & 119 & 115 & 121 & 128 & 146 & 278 \\
& B & 224 & 229 & 233 & 238 & 250 & 247 & 253 & 263 & 267 & 277 & 294 & 287 & 291 & 317 & 288 & 273 & 252 \\ \hline
$\uparrow q$

%% file: spacetable
&& $2^{10}$& $2^{11}$& $2^{12}$& $2^{13}$& $2^{14}$& $2^{15}$& $2^{16}$& $2^{17}$& $2^{18}$& $2^{19}$& $2^{20}$& $2^{21}$& $2^{22}$& $2^{23}$& $2^{24}$& $2^{25}$& $2^{26}$& $\leftarrow \ell$ \\ \hline
$2^{10}$ & T2 & 2 \\
& V & 2 \\
& B & 2Gb \\ \hline
$2^{11}$ & T2 & 2 & 2 \\
& V & 2 & 2 \\
& B & 2Gb & 2Gb \\ \hline
$2^{12}$ & T2 & 2 & 2 & 2 \\
& V & 2 & 2 & 2 \\
& B & 2Gb & 2Gb & 2Gb \\ \hline
$2^{13}$ & T2 & 2 & 2 & 2 & 2 \\
& V & 2 & 2 & 2 & 2 \\
& B & 2Gb & 2Gb & 2Gb & 2Gb \\ \hline
$2^{14}$ & T2 & 2 & 2 & 2 & 3 & 3 \\
& V & 2 & 2 & 2 & 2 & 2 \\
& B & 2Gb & 2Gb & 2Gb & 2Gb & 2Gb \\ \hline
$2^{15}$ & T2 & 2 & 2 & 2 & 3 & 4 & 3 \\
& V & 2 & 2 & 2 & 2 & 2 & 2 \\
& B & 2Gb & 2Gb & 2Gb & 2Gb & 2Gb & 2Gb \\ \hline
$2^{16}$ & T2 & 2 & 2 & 2 & 3 & 4 & 6 & 5 \\
& V & 3 & 3 & 3 & 3 & 3 & 3 & 3 \\
& B & 2Gb & 2Gb & 2Gb & 2Gb & 2Gb & 2Gb & 2Gb \\ \hline
$2^{17}$ & T2 & 2 & 2 & 2 & 3 & 4 & 6 & 10 & 8 \\
& V & 4 & 4 & 4 & 4 & 4 & 4 & 4 & 4 \\
& B & 2Gb & 2Gb & 2Gb & 2Gb & 2Gb & 2Gb & 2Gb & 2Gb \\ \hline
$2^{18}$ & T2 & 2 & 2 & 2 & 3 & 5 & 7 & 10 & 14 & 14 \\
& V & 6 & 6 & 6 & 6 & 6 & 6 & 6 & 6 & 7 \\
& B & 2Gb & 2Gb & 2Gb & 2Gb & 2Gb & 2Gb & 2Gb & 2Gb & 2Gb \\ \hline
$2^{19}$ & T2 & 2 & 2 & 2 & 3 & 5 & 8 & 12 & 20 & 26 & 26 \\
& V & 11 & 11 & 11 & 11 & 11 & 11 & 11 & 11 & 11 & 11 \\
& B & 2Gb & 2Gb & 2Gb & 2Gb & 2Gb & 2Gb & 2Gb & 2Gb & 2Gb & 2Gb \\ \hline
$2^{20}$ & T2 & 2 & 2 & 2 & 3 & 5 & 8 & 12 & 22 & 34 & 50 & 50 \\
& V & 21 & 21 & 21 & 21 & 21 & 21 & 21 & 21 & 21 & 21 & 21 \\
& B & 2Gb & 2Gb & 2Gb & 2Gb & 2Gb & 2Gb & 2Gb & 2Gb & 2Gb & 2Gb & 2Gb \\ \hline
$2^{21}$ & T2 & 2 & 2 & 2 & 3 & 5 & 8 & 14 & 22 & 38 & 66 & 98 & 98 \\
& V & 41 & 41 & 41 & 41 & 41 & 41 & 42 & 41 & 41 & 41 & 42 & 41 \\
& B & 2Gb & 2Gb & 2Gb & 2Gb & 2Gb & 2Gb & 2Gb & 2Gb & 2Gb & 2Gb & 2Gb & 2Gb \\ \hline
$2^{22}$ & T2 & 2 & 2 & 2 & 4 & 5 & 8 & 15 & 25 & 42 & 82 & 130 & 194 & 194 \\
& V & 81 & 81 & 81 & 81 & 81 & 81 & 82 & 81 & 81 & 81 & 81 & 81 & 81 \\
& B & 3Gb & 3Gb & 3Gb & 3Gb & 3Gb & 3Gb & 3Gb & 3Gb & 3Gb & 3Gb & 3Gb & 3Gb & 3Gb \\ \hline
$2^{23}$ & T2 & 2 & 2 & 3 & 4 & 5 & 8 & 14 & 27 & 46 & 97 & 161 & 257 & 385 & 386 \\
& V & 161 & 161 & 161 & 161 & 161 & 161 & 161 & 162 & 161 & 161 & 161 & 161 & 161 & 162 \\
& B & 4Gb & 4Gb & 4Gb & 4Gb & 4Gb & 4Gb & 4Gb & 4Gb & 4Gb & 4Gb & 4Gb & 4Gb & 4Gb & 3Gb \\ \hline
$2^{24}$ & T2 & 3 & 2 & 3 & 3 & 5 & 9 & 15 & 27 & 53 & 90 & 160 & 318 & 510 & 766 & 770 \\
& V & 321 & 320 & 321 & 322 & 320 & 320 & 320 & 320 & 319 & 319 & 320 & 319 & 319 & 319 & 320 \\
& B & 6Gb & 6Gb & 6Gb & 6Gb & 6Gb & 6Gb & 6Gb & 6Gb & 6Gb & 5Gb & 5Gb & 5Gb & 5Gb & 5Gb & 5Gb \\ \hline
$2^{25}$ & T2 & 3 & 3 & 3 & 4 & 5 & 8 & 15 & 29 & 53 & 96 & 158 & 314 & 506 & 1011 & 1Gb & 2Gb \\
& V & 634 & 634 & 634 & 634 & 634 & 634 & 634 & 634 & 634 & 634 & 634 & 634 & 634 & 634 & 634 & 634 \\
& B & 9Gb & 9Gb & 9Gb & 9Gb & 9Gb & 9Gb & 9Gb & 9Gb & 9Gb & 9Gb & 9Gb & 9Gb & 8Gb & 8Gb & 8Gb & 7Gb \\ \hline
$2^{26}$ & T2 & 5 & 4 & 3 & 4 & 5 & 9 & 17 & 27 & 55 & 109 & 169 & 307 & 499 & 996 & 2Gb & 3Gb & 3Gb \\
& V & 1Gb & 1Gb & 1Gb & 1Gb & 1Gb & 1Gb & 1Gb & 1Gb & 1Gb & 1Gb & 1Gb & 1Gb & 1Gb & 1Gb & 1Gb & 1Gb & 1Gb \\
& B & 14Gb & 14Gb & 14Gb & 14Gb & 14Gb & 14Gb & 14Gb & 14Gb & 14Gb & 14Gb & 14Gb & 14Gb & 14Gb & 14Gb & 14Gb & 13Gb & 12Gb \\ \hline
$\uparrow q$